\documentclass[a4paper,11pt]{article}

\usepackage[utf8]{inputenc}
\usepackage[T1]{fontenc}
\usepackage{amsmath}
\usepackage{amsthm}
\usepackage{microtype}
\usepackage{fullpage}
\usepackage{amssymb}
\usepackage[hidelinks]{hyperref}
\usepackage[nameinlink]{cleveref}
\usepackage{centernot}
\usepackage{dsfont}

\theoremstyle{plain}
\newtheorem{theorem}{Theorem}
\theoremstyle{definition}

\theoremstyle{remark}

\theoremstyle{plain}
\newtheorem{lemma}[theorem]{Lemma}

\newtheorem{claim}[theorem]{Claim}

\newcommand{\prn}[1]{\left(#1\right)}
\newcommand{\cprn}[1]{\!\left(#1\right)}
\newcommand{\sqbra}[1]{\left[#1\right]}

\newcommand{\abs}[1]{\left|#1\right|}

\newcommand{\brkts}[1]{\left\{#1\right\}}

\newcommand{\bool}{\brkts{0,1}}

\newcommand{\TV}{\mathrm{TV}}
\newcommand{\Bern}{\mathrm{Bern}}
\newcommand{\PMFEquals}{\textsc{PMFEquals}}

\newcommand{\SubsetProd}{\textsc{SubsetProd}}

\newcommand{\SZK}{\mathsf{SZK}}
\newcommand{\NISZK}{\mathsf{NISZK}}

\newcommand{\tv}{d_{\mathrm{TV}}}

\newcommand{\dtv}{\tv}
\newcommand{\ignore}[1]{}

\renewcommand{\P}{\mathsf{P}}

\allowdisplaybreaks

\title{\bf Total Variation Distance for Product Distributions is $\#\P$-Complete%
\thanks{An extended version of this paper appeared in the proceedings of IJCAI 2023~\cite{bhattacharyya2023approximating}.}}

\author{%
Arnab Bhattacharyya \\
National University of Singapore
\and
Sutanu Gayen \\
IIT Kanpur
\and
Kuldeep S. Meel \\
University of Toronto
\and
Dimitrios Myrisiotis \\
CNRS@CREATE LTD.
\and
A. Pavan \\
Iowa State University
\and
N.~V.~Vinodchandran \\
University of Nebraska-Lincoln
}
\date{}
\begin{document}

\maketitle

\begin{abstract}
We show that computing the total variation distance between two product distributions is $\#\P$-complete.
This is in stark contrast with other distance measures such as Kullback–Leibler, Chi-square, and Hellinger, which tensorize over the marginals leading to efficient algorithms.
\end{abstract}

\section{Introduction}

\label{sec:introduction}

Given two distributions $P$ and $Q$ over a finite domain $\mathcal{D}$, their {\em total variation (TV) distance} or {\em statistical difference} $\dtv(P,Q)$ is defined as
\[
\dtv(P,Q)
= \max_{S \subseteq \mathcal{D}} \cprn{P(S)-Q(S)}
= \frac12 \sum_{x \in \mathcal{D}} |P(x) - Q(x)|
= \sum_{x \in \mathcal{D}} \max \cprn{0, P(x) - Q(x)}.
\]
Given two distributions $P$ and $Q$ over a finite domain $\mathcal{D}$, how hard is it to compute $\dtv(P,Q)$?
If $P$ and $Q$ are explicitly specified by the probabilities of all of the points of the (discrete) domain $\mathcal{D}$, summing up the absolute values of the differences in probabilities at all points leads to a simple linear time algorithm.
However, in many applications, the distributions of interest are of a high dimension with succinct representations.
In these scenarios, since the size of the domain $\mathcal{D}$ is exponentially large compared to its representation, an $O(|\mathcal{D}|)$ algorithm is highly impractical.
This leads to the following fundamental computational question:
\begin{quote}
\it
How hard is it to compute TV distance for high dimensional distributions?
\end{quote}
The simplest model for a high-dimensional distribution is the {\em product distribution}, which is a product of independent Bernoulli distributions.
Let a Bernoulli distribution with parameter $p$ be denoted by $\Bern\cprn{p}$.
A product distribution $P$ over $\{0,1\}^n$ can be described by $n$ parameters $p_1,\dots,p_n$ where each $p_i \in [0,1]$ is the probability that the $i$-th coordinate equals $1$ (such a $P$ is usually denoted by $\bigotimes_{i=1}^n\Bern(p_i)$).
For any $x\in \{0,1\}^n$, the probability of $x$ with respect to the distribution $P$ is given by $P\cprn{x}=\prod_{i\in S}p_i\prod_{i\in\sqbra{n}\setminus S}\prn{1-p_i}\in\sqbra{0,1}$, where $S \subseteq \sqbra{n}$ is such that $i\in S$ if and only if the $i$-th coordinate of $x$ is $1$, independently.

We show that the exact computation of the total variation distance between two product distributions $P$ and $Q$ is $\#\P$-complete.
This is a surprising result, given that for many other distance measures such as Hellinger, Chi-square, and Kullback–Leibler divergence, there are efficient algorithms for computing the distance between two product distributions.

\paragraph{Notation:}

We use $\sqbra{n}$ to denote the \emph{ordered} set $\brkts{1,\dots,n}$ and $\log$ to denote $\log_2$.
Moreover, we shall assume that all probabilities are represented as rational numbers of the form $a/b$.

\section{The Hardness of Computing TV Distance}

\label{sec:hardness}

\begin{theorem}
\label{thm:hardness-Bern-products-intro}
Given two product distributions $P$ and $Q$, computing $\dtv(P,Q)$ is $\#\P$-complete.
\end{theorem}

We first introduce a problem called $\#\PMFEquals$ and show that it is $\#\P$-hard.
We have that $\#\PMFEquals$ is the following problem:
Given a probability vector $\prn{p_1,\dots,p_n}$ where $p_i\in\sqbra{0,1}$ and a number $v$, compute the number of $x\in\bool^n$ such that $P\cprn{x}=v$ (that is the number $|\{x \in \{0,1\}^n\mid P(x)=v\}|$), where $P$ is the product distribution described by $\prn{p_1,\dots,p_n}$.

\begin{lemma}
\label{lemma:PMFEquals-is-hard}
{\rm $\#\PMFEquals$} is $\#\P$-hard.
\end{lemma}

\begin{proof}
We will show a reduction from $\#\SubsetProd$ to $\#\PMFEquals$, and the result will follow from the fact that $\#\SubsetProd$ is $\#\P$-hard.
$\#\SubsetProd$ is the following problem:
Given integers $a_1,\dots,a_n,T$, compute the number of sets $S\subseteq\sqbra{n}$ such that $\prod_{i\in S}a_i=T$.

Let $a_1,\dots,a_n$, and $T$ be the numbers of an arbitrary $\#\SubsetProd$ instance $I_S$.
We will create a $\#\PMFEquals$ instance $I_P$ that has the same number of solutions as $I_S$.
To this end, let $p_i:=\frac{{a_i}}{1+{a_i}}$ for every $i$ and $v:=T\prod_{i\in\sqbra{n}}(1-p_i)$, and observe that $a_i=\frac{p_i}{1-p_i}$.
For any set $S\subseteq [n]$, we have the following equivalences:
\[
\prod_{i\in S}a_i=T
\Leftrightarrow\prod_{i\in S}\frac{p_i}{1-p_i}=\frac{v}{\prod_{i\in\sqbra{n}}(1-p_i)}
\Leftrightarrow\prod_{i\in S}p_i\prod_{i\notin S}\prn{1-p_i}=v
\Leftrightarrow P\cprn{x_S}=v,
\]
where $x_S$ is such that $x_i=1$ if and only if $i\in S$.
This completes the proof.
\end{proof}

We now turn to \Cref{thm:hardness-Bern-products-intro}.

\begin{proof}[Proof of \Cref{thm:hardness-Bern-products-intro}]
We separately prove membership in $\#\P$ and $\#\P$-hardness.

\paragraph{Membership in $\#\P$:}

Let $P$ and $Q$ be two product distributions specified by $p_1,\dots,p_n$ and $q_1,\dots,q_n$, respectively.
Without loss of generality we shall assume that these parameters are fractions (as we only have some finite precision available).
The goal is to design a nondeterministic machine ${\cal N}$ that takes $p_1,\dots,p_n$ and $q_1,\dots,q_n$ as inputs and is such that the number of its accepting paths (normalized by an appropriate quantity) equals $\dtv(P,Q)$.

Let $M$ be the product of the denominators of all parameters $p_1,\dots,p_n,q_1,\dots,q_n$ and their complements $1-p_1,\dots,1-p_n,1-q_1,\dots,1-q_n$.
The non-deterministic machine ${\cal N}$ first guesses an element $i\in\bool^n$ in the sample space of $P$ and $Q$, computes $|P(i)-Q(i)|$ by using the parameters $p_1,\dots,p_n,q_1,\dots,q_n$, then guesses an integer $0\leq z\leq M$, and finally accepts if and only if $1\leq z\leq M|P(i)-Q(i)|$.
(Note that $M|P(i)-Q(i)|=\abs{M\cdot P(i)-M\cdot Q(i)}$ is an integer.)

It follows that $\dtv(P,Q) = \frac{1}{2}\sum_{i\in\bool^n}\abs{P\cprn{i}-Q\cprn{i}}$ is equal to the number of accepting paths of ${\cal N}$ divided by ${2M}$.

\paragraph{$\#\P$-hardness:}
 
For establishing hardness, we will show a reduction from $\#\PMFEquals$ to computing TV distance between product distributions.
The theorem will then follow from \Cref{lemma:PMFEquals-is-hard}.

Let $p_1,\dots,p_n$ and $v$ be the numbers in an arbitrary instance of $\#\PMFEquals$ where each $p_i$ is represented as an $m$-bit binary fraction.
With this, $P(x)$ can be represented as an $nm$-bit binary fraction.
Thus without loss of generality, we can assume that $v$ is also an $nm$-bit fraction.
We distinguish between two cases depending on whether $v<2^{-n}$ or $v\geq2^{-n}$.

\subparagraph{Case A: $v<2^{-n}$.}

First, we construct two distributions $\hat P=\Bern\cprn{\hat p_1}\otimes\cdots\otimes \Bern\cprn{\hat p_{n+1}}$ and $\hat Q=\Bern\cprn{\hat q_1}\otimes\cdots\otimes\Bern\cprn{\hat q_{n+1}}$ over $\{0,1\}^{n+1}$ as follows:
$\hat p_i:=p_i$ for $i\in\sqbra{n}$ and $\hat p_{n+1}:=1$; $\hat q_i:=1/2$ for $i\in\sqbra{n}$ and $\hat q_{n+1}:=v2^n$.
We have that $\dtv\cprn{\hat P,\hat Q}$ is, by definition, equal to the sum $\sum_{x\in\{0,1\}^{n+1}}\max\cprn{0,\hat P\cprn{x}-\hat Q\cprn{x}}$ or
\begin{align}
\sum_{x\in\bool^n}\max\cprn{0,P\cprn{x}-\frac{1}{2^n}v2^n}
=\sum_{x\in\bool^n}\max\cprn{0,P\cprn{x}-v}
=\sum_{x:P\prn{x}>v}\prn{P\cprn{x}-v}.
\label{eq:hardness}
\end{align}
We define two more distributions $P'$ and $Q'$ over $\{0,1\}^{n+2}$, by making use of the following claim.

\begin{claim}
\label{clm:beta}
For $\beta = \frac{1}{2^{3nm}}$, the following hold for all $x$:
If $P\cprn{x}<v$, then $P\cprn{x}\prn{\frac{1}{2}+\beta}<v\prn{\frac{1}{2}-\beta}$;
if $P\cprn{x}>v$, then $P\cprn{x}\prn{\frac{1}{2}-\beta}>v\prn{\frac{1}{2}+\beta}$.
\end{claim}

\begin{proof}
For \Cref{clm:beta} to hold, observe that we want $\beta$ to be at most $\frac{|v-P(x)|}{v+P(x)}$ for every $x$, so that $P(x) \neq v$.
Since both $v$ and $P(x)$ have $nm$-bit representations, both $|v-P(x))|$ and $v+P(x)$ have $nm$-bit representations. 
Thus $\frac{|v-P(x)|}{v+P(x)}$ can be represented as a $2nm$-bit fraction.
Since this fraction is not zero, and the smallest $2nm$-bit fraction is $\frac{1}{2^{2nm}}$, choosing $\beta := \frac{1}{2^{3nm}}$ suffices.
\end{proof}

We now define $P'$ and $Q'$ as follows:
$p_i':=p_i$ for $i\in [n]$, $p_{n+1}':=1$, and $p_{n+2}':=\frac{1}{2}+\beta$; $q_i':=\frac{1}{2}$ for $i\in [n]$, $q_{n+1}':=v2^n$, and $q_{n+2}':=\frac{1}{2}-\beta$ where $\beta$ is as in \Cref{clm:beta}.
We establish the following.

\begin{claim}
\label{clm:v-small}
We have that $\abs{\brkts{x\mid P\cprn{x}=v}} = \frac{1}{2\beta v}\prn{d_{\TV}\cprn{P',Q'}-d_{\TV}\cprn{\hat P,\hat Q}}$.
\end{claim}

\begin{proof}
We have that $d_{\TV}\cprn{P',Q'}$ is equal to $\sum_{x\in \{0, 1\}^{n+2}}\max\cprn{0,P'\cprn{x}-Q'\cprn{x}}$ or
\begin{align*}
&\sum_{x\in\{0,1\}^n}\max\cprn{0,P\cprn{x}\prn{\frac{1}{2}+\beta}-\frac{v2^n}{2^n}\prn{\frac{1}{2}-\beta}}+\max\cprn{0,P\cprn{x}\prn{\frac{1}{2}-\beta}-\frac{v2^n}{2^n}\prn{\frac{1}{2}+\beta}}\\
&=\sum_{x}\max\cprn{0,P\cprn{x}\prn{\frac{1}{2}+\beta}-v\prn{\frac{1}{2}-\beta}}
+\max\cprn{0,P\cprn{x}\prn{\frac{1}{2}-\beta}-v\prn{\frac{1}{2}+\beta}}\\
&=\sum_{x:P\prn{x}\geq v} P\cprn{x}\prn{\frac{1}{2}+\beta}-v\prn{\frac{1}{2}-\beta} +\sum_{x:P\prn{x}>v}P\cprn{x}\prn{\frac{1}{2}-\beta}-v\prn{\frac{1}{2}+\beta}\\
&=\sum_{x:P\prn{x}= v}P\cprn{x}\prn{\frac{1}{2}+\beta}-v\prn{\frac{1}{2}-\beta}+\sum_{x:P\prn{x}> v} P\cprn{x}\prn{\frac{1}{2}+\beta}-v\prn{\frac{1}{2}-\beta}\\
&\qquad+\sum_{x:P\prn{x}>v}P\cprn{x}\prn{\frac{1}{2}-\beta}-v\prn{\frac{1}{2}+\beta}
=2\beta v\abs{\brkts{x\mid P\cprn{x}=v}}
+\sum_{x:P\prn{x}>v}\prn{P\cprn{x}-v}.
\end{align*}
The result now follows from \Cref{eq:hardness}.
The first equality above follows from the definitions of $P'$ and $Q'$ (since $p_{n+1}'= 1$).
Note that for every $x$ with $P(x) < v$, by Claim~\ref{clm:beta}, $P(x)(\frac{1}{2}+\beta) < v(\frac{1}{2}-\beta)$ and if $P(x) \geq v$, then $P(x)(\frac{1}{2}+\beta) \geq  v(\frac{1}{2}-\beta)$.
Also when $P(x) \leq  v$, we have that $P(x)(\frac{1}{2}+\beta)$ is at most $v(\frac{1}{2}-\beta)$ and when $P(x) > v$, by Claim~\ref{clm:beta}, we have $P(x)(\frac{1}{2}-\beta) > v(\frac{1}{2}+\beta)$.
These imply the second equality.
The rest of the equalities hold by algebraic manipulations.
\end{proof}

Thus, $|\{x\mid P(x) = v\}|$ can be computed by 
computing $\dtv(P',Q')$ and $\dtv(\hat{P},\hat{Q})$. 
Thus the proof in this case follows by \Cref{lemma:PMFEquals-is-hard}.

\subparagraph{Case B: $v \geq 2^{-n}$.}

First, let us define distributions $\hat P=\Bern\cprn{\hat p_1}\otimes\cdots\otimes\Bern\cprn{\hat p_n}$ and $\hat Q=\Bern\cprn{\hat q_1}\otimes\cdots\otimes\Bern\cprn{\hat q_n}$ as follows: $\hat p_i:=p_i$ for $i\in\sqbra{n}$ and $\hat p_{n+1}:=\frac{1}{v2^n}$; $\hat q_i:=\frac{1}{2}$ for $i\in\sqbra{n}$ and $\hat q_{n+1}:=1$.

We now have that $d_{\TV}\cprn{\hat P,\hat Q}$ is equal to $\frac{1}{2}\sum_{x}\abs{\hat P\cprn{x}-\hat Q\cprn{x}}$ or
\begin{align*}
\sum_{x}\max\cprn{0,\hat P\cprn{x}-\hat Q\cprn{x}}
&=\sum_{x}\max\cprn{0,P\cprn{x}\frac{1}{v2^n}-\frac{1}{2^n}}
+\sum_{x}\max\cprn{0,P\cprn{x}\prn{1-\frac{1}{v2^n}}}\\
&=\sum_{x}\max\cprn{0,P\cprn{x}\frac{1}{v2^n}-\frac{1}{2^n}}
+1-\frac{1}{v2^n}.
\end{align*}
As earlier, we define two more distributions $P'$ and $Q'$, by making use of \Cref{clm:beta}.
The new distributions $P'$ and $Q'$ are such that $p_i':=p_i$ for $i\in\sqbra{n}$, $p_{n+1}':=\frac{1}{v2^n}$, and $p_{n+2}':=\frac{1}{2}+\beta$; $q_i':=1/2$ for $i\in\sqbra{n}$, $q_{n+1}':=1$, and $q_{n+2}':=\frac{1}{2}-\beta$.
We establish the following claim.

\begin{claim}
\label{clm:v-big}
We have that $\abs{\brkts{x\mid P\cprn{x}=v}} = \frac{2^{n-1}}{\beta}\prn{d_{\TV}\cprn{P',Q'}-d_{\TV}\cprn{\hat P,\hat Q}}$.
\end{claim}

\begin{proof}
We have that $d_{\TV}\cprn{P',Q'}$ is equal to $\sum_{x}\max\cprn{0,P'\cprn{x}-Q'\cprn{x}}$ or
\begin{align*}
&\sum_{x}\max\cprn{0,P\cprn{x}\frac{1}{v2^n}\prn{\frac{1}{2}+\beta}-\frac{1}{2^n}\prn{\frac{1}{2}-\beta}}
+\max\cprn{\!0,P\cprn{x}\frac{1}{v2^n}\prn{\frac{1}{2}-\beta}-\frac{1}{2^n}\prn{\frac{1}{2}+\beta}}\\
&\qquad+\sum_{x}\max\cprn{0,P\cprn{x}\prn{1-\frac{1}{v2^n}}\prn{\frac{1}{2}+\beta}}
+\max\cprn{0,P\cprn{x}\prn{1-\frac{1}{v2^n}}\prn{\frac{1}{2}-\beta}}\\
&=\sum_{x}\max\cprn{0,P\cprn{x}\frac{1}{v2^n}\prn{\frac{1}{2}+\beta}-\frac{1}{2^n}\prn{\frac{1}{2}-\beta}}\\
&\qquad+\sum_{x}\max\cprn{0,P\cprn{x}\frac{1}{v2^n}\prn{\frac{1}{2}-\beta}-\frac{1}{2^n}\prn{\frac{1}{2}+\beta}}
+ \max\cprn{0,P\cprn{x}\prn{1-\frac{1}{v2^n}}}.
\end{align*}
Applying \Cref{clm:beta}, we have that $d_{\TV}\cprn{P',Q'}$ is equal to
\begin{align*}
&\sum_{x:P\prn{x}\geq v}\max\cprn{0,P\cprn{x}\frac{1}{v2^n}\prn{\frac{1}{2}+\beta}-\frac{1}{2^n}\prn{\frac{1}{2}-\beta}}\\
&\qquad+\sum_{x:P\prn{x}>v}\max\cprn{0,P\cprn{x}\frac{1}{v2^n}\prn{\frac{1}{2}-\beta}-\frac{1}{2^n}\prn{\frac{1}{2}+\beta}}
+1-\frac{1}{v2^n}\\
&=\sum_{x:P\prn{x}=v}\max\cprn{0,P\cprn{x}\frac{1}{v2^n}\prn{\frac{1}{2}+\beta}-\frac{1}{2^n}\prn{\frac{1}{2}-\beta}}\\
&\qquad+\sum_{x:P\prn{x}>v}\max\cprn{0,P\cprn{x}\frac{1}{v2^n}\prn{\frac{1}{2}+\beta}-\frac{1}{2^n}\prn{\frac{1}{2}-\beta}}\\
&\qquad+\sum_{x:P\prn{x}>v}\max\cprn{0,P\cprn{x}\frac{1}{v2^n}\prn{\frac{1}{2}-\beta}-\frac{1}{2^n}\prn{\frac{1}{2}+\beta}}
+1-\frac{1}{v2^n}\\
&=2\beta\frac{1}{2^n}\abs{\brkts{x\mid P\cprn{x}=v}}
+\sum_{x}\max\cprn{0,P\cprn{x}\frac{1}{v2^n}-\frac{1}{2^n}}
+1-\frac{1}{v2^n}\\
&=\frac{\beta}{2^{n-1}}\abs{\brkts{x\mid P\cprn{x}=v}}+d_{\TV}\cprn{\hat P,\hat Q}.
\end{align*}
This concludes the proof.
\end{proof}

Thus, also in this case, the proof follows by \Cref{lemma:PMFEquals-is-hard}.
\end{proof}

\section{Bibliographical Remarks}

\label{sec:relatedwork}

\cite{SV03} established, in a seminal work, that additively approximating the TV distance between two distributions that are samplable by Boolean circuits is hard for $\SZK$.
Subsequent works captured variations of this theme~\cite{GoldreichSV99,Malka15,DPV20}:
For example,~\cite{GoldreichSV99} showed that the problem of deciding whether a distribution samplable by a Boolean circuit is close or far from the uniform distribution is complete for $\NISZK$.
Finally, \cite{CortesMR07,LyngsoP02,hmm18} show that it is undecidable to check whether the TV distance between two hidden Markov models is greater than a threshold or not, and that it is $\#\P$-hard to additively approximate it.

On an algorithmic note, \cite{0001GMV20} designed efficient algorithms to additively approximate the TV distance between distributions efficiently samplable and efficiently computable.
In a similar vein, \cite{PM21} studied a related \emph{property testing} variant of TV distance for distributions encoded by circuits.
Recently, \cite{feng2023simple} designed an FPRAS for relatively approximating the TV distance between two arbitrary product distributions and \cite{feng2024deterministically} gave an FPTAS for it.

\paragraph{Acknowledgements:}

This work was supported in part by the National Research Foundation Singapore under its NRF Fellowship Programme [NRF-NRFFAI-2019-0002, NRF-NRFFAI1-2019-0004], Amazon Faculty Research Awards, an initiation grant by IIT Kanpur (SG), SERB award CRG/2022/007985 (SG), and NSF awards 2130536 and 2130608.
This work was done in part when AB, DM, AP, and NVV visited the Simons Institute for the Theory of Computing.

% {\small
% \bibliographystyle{alpha}
% \bibliography{references}}

{\small
\newcommand{\etalchar}[1]{$^{#1}$}

}

\end{document}